\documentclass[twocolumn, switch]{article}
\pdfoutput=1
\usepackage{preprint}
\usepackage{hyperref}
\usepackage[numbers,square]{natbib}
\usepackage{amssymb}

\usepackage[nogroupskip,toc,acronym]{glossaries-extra}
\setabbreviationstyle[acronym]{long-short}
\usepackage{svg}       

\hypersetup{colorlinks=true, linkcolor=black, urlcolor=black, citecolor=black, anchorcolor=black}

\usepackage{enumitem}
\usepackage[utf8]{inputenc}	
\usepackage[T1]{fontenc}
\PassOptionsToPackage{dvipsnames}{xcolor}
\usepackage[dvipsnames]{xcolor}

\newacronym{ETC}{ETC}{Event-triggered Control}
\newacronym{TTC}{TTC}{Time-triggered Control}
\newacronym{KO}{KO}{Koopman Operator}
\newacronym{LMI}{LMI}{Linear Matrix Inequality}
\newacronym{LTI}{LTI}{Linear Time Invariant}
\newacronym{KOETC}{KOETC}{Koopman Operator-Based Event-Triggered Control}
\newacronym{RBF}{RBF}{Radial Basis Function}
\usepackage{newfloat}
\usepackage{sidecap}
\sidecaptionvpos{figure}{c}

\usepackage{titlesec}
\titlespacing\section{0pt}{12pt plus 3pt minus 3pt}{1pt plus 1pt minus 1pt}
\titlespacing\subsection{0pt}{10pt plus 3pt minus 3pt}{1pt plus 1pt minus 1pt}
\titlespacing\subsubsection{0pt}{8pt plus 3pt minus 3pt}{1pt plus 1pt minus 1pt}

\definecolor{lime}{HTML}{A6CE39}

\title{Koopman-Based Event-Triggered Control from Data}

\usepackage{authblk}

\author[1,2]{Zeyad M. Manaa}
\author[1,2]{Ayman M. Abdallah}
\author[1,2]{Mohamed Ismail}
\author[3,4]{Samil El Ferik}

\affil[1]{Department of Aerospace Engineering, King Fahd University of Petroleum and Minerals (KFUPM), Dhahran, Saudi Arabia}
\affil[2]{Interdisciplinary Research Center for Aviation and Space Exploration, KFUPM}
\affil[3]{Department of Control and Instrumentation Engineering, KFUPM}
\affil[4]{Interdisciplinary Research Center for Smart Mobility and Logistics, KFUPM}

\begin{document}

\twocolumn[\begin{@twocolumnfalse}

\maketitle

\begin{abstract}
\gls{ETC} presents a promising paradigm for efficient resource usage in networked and embedded control systems by reducing communication instances compared to traditional time-triggered strategies. This paper introduces a novel approach to \gls{ETC} for discrete-time nonlinear systems using a data-driven framework. By leveraging Koopman operator theory, the nonlinear system dynamics are globally linearized (approximately in practical settings) in a higher-dimensional space. We design a state-feedback controller and an event-triggering policy directly from data, ensuring exponential stability in Lyapunov sense. The proposed method is validated through extensive simulation experiments, demonstrating significant resource savings.
\end{abstract}


\vspace{0.5cm}

\end{@twocolumnfalse}]


\section{Introduction}
\gls{ETC} is an implementation strategy in which the plant and its controller only exchange data when certain output- or state-related conditions are met. Event-triggered control seeks to reduce communication instances by concentrating on the real needs of the system. This contrasts with traditional, conservative time-triggered strategies that depend on fixed communication intervals. In situations where efficient use of resources is essential, such as networked and embedded control systems, this paradigm has gained increasing attention. \gls{ETC} strategies, which offer improved system performance and resource savings in a variety of setups and control problems, have been developed in the literature thanks to the early results in Refs. \cite{aaarzen1999simple, eker2000feedback, tabuada2007event, heemels2012introduction}.

Parametric state-space models are the foundation of traditional control engineering literature, where the plant to be controlled must be identified modeled or identified firstly. These models use system data and are often derived from first principles or architecturally constrained system identification techniques. But in cases when first-principles models are intricate or hard to derive, they can only be considered as approximate representations of real systems, which inevitably leads to modeling errors. These errors impede accurate control design. They propagate through the analysis and implementation phases, ultimately degrading overall system performance.

By excluding the demand for explicit system identification and instead of leveraging data gathered from open-loop simulations/experiments for any system control analysis and design, data-driven control techniques serve as a promising alternative. Several data-driven techniques for creating state feedback controllers and illustrating system dynamics have been shown in recent works, such as those by da Silva et al. \cite{da2018data}, and De Persis and Tesi \cite{de2019formulas}. These techniques greatly streamline the control design process and do not require constantly exciting input data. There are also numerous applications of data-driven control in fields such as robotics \cite{shi2021enhancement}, aerospace \cite{folkestad2021quadrotor}, and power systems \cite{susuki2016applied}. { Other methods, when the model is completely unknown, such as SINDy \cite{brunton2016discovering} can be utilized to firstly get a nonlinear representation of the dynamics of the system. For example, this approach is applied to model the dynamics of: i) quadrotors'  \cite{manaa2024data}, ii) disease \cite{jiang_modeling_2021}, iii) optics communication systems \cite{sorokina_sparse_2016}, iv) chemical processes \cite{bhadriraju_operable_2020}, v) and robotics applications \cite{bhattacharya_sparse_2020}.
    }

Alternatively, when first-principles or system identification methods fail, the controller can be constructed directly using the input, state/output data that are accessible. This approach, referred to as direct data-driven control, \cite{sznaier2020control, campi2002virtual, fliess2013model} constructs controllers directly from data. Although the literature is full of data-driven techniques for control,  only a limited number of techniques exist in the current literature \cite{liu2023data, digge2022data, wang2023model} for data-driven event-based control, particularly for nonlinear systems.
Hence, there is a strong demand for comprehensive data-driven event-based control methods tailored for general nonlinear systems, particularly applicable to discrete-time systems in our case. In many cases, it is appropriate and feasible to formulate the control and triggering conditions as data-dependent \gls{LMI}. Given that most of the existing literature on \gls{ETC} is well developed for \gls{LTI} systems, we aim to globally linearize nonlinear systems by increasing the dimensional space in which they reside.

 
 This is not entirely a new idea. In the 1930s, Koopman and von Neumann~\cite{koopman1931hamiltonian, koopman1932dynamical} introduced a trade-off between the nonlinear nature of dynamical systems and their infinite-dimensional representations, which appear linear in the lifted space. Another resurgence of attention in mid 2000s in the work of Mezi\'c and Banaszuk \cite{mezic2004comparison, mezic2005spectral} has led to new applications and studies using the idea in many fields including, robotics, fluid dynamics, epidemiology, \cite{bruder2021koopman, mamakoukas2021derivative, zhu2022koopman, komeno2022deep, han2021desko, manaa2024koopman, hossain2023data, markmann2024koopman, mezic2024koopman} and many other fields due to the intersection between data science and the easy-to-access computational domain.

We consequently propose \gls{KOETC}, a technique inspired by \gls{KO} to acquire (approximately) global linear systems but in a higher dimensional space. Afterwards, we design the controller and the triggering policy for \gls{ETC} for discrete-time linear systems directly from controlled system data, all together ensuring performance metrics (i.e. Lyapunov exponential stability).

\subsection{Contributions}
By combining Koopman operator theory with event-triggered control (ETC), this paper makes a contribution by introducing a Koopman-based approach to ETC for discrete-time nonlinear systems. Through the approximate global linearization of nonlinear dynamics made possible by this integration, the following direct, data-driven designs are made possible:
\begin{enumerate}[label=\roman*)]
    \item An event-triggering policy minimizes resource consumption by updating control actions only when required, reducing communication instances; and
    \item A state-feedback controller, which effectively stabilizes the system by utilizing the Koopman-lifted linear dynamics.
\end{enumerate}
In comparison to time-triggered approaches, the \gls{KOETC} framework reduces communication events in simulations while achieving stability in the Lyapunov sense.

The rest of this paper is structured to methodically construct and validate the suggested KOETC framework after the motivation and goals have been established. The preliminary information and notations that are necessary to comprehend our methodology are outlined in Section 2. In Section 3, the KOETC framework's design is examined in detail, including the triggering policy and data-driven controller. We provide simulation results in Section 4, which show how effective the approach is. Section 5 provide additional discussions on practical guidelines for the proposed method. Finally, Section 6 concludes the paper and outlines potential directions for future research.

\section{Preliminaries}
\subsection{Notations and Basic Definitions}
Let $\bbZ_{\geq 0} := \{0, 1, 2, \dots\}$ denote the set of nonnegative integers, and let $\bbZ_{> 0} := \bbZ_{\geq 0} \setminus \{0\}$ denote the set of positive integers. We denote by $\bbR$ the set of real numbers and use a similar notation as for $\bbZ$. The $\ell_2$ norm of a vector (a finite sequence) is denoted by $\|\cdot\|$. The symbols $I$ and $\underbar{0}$ denote the identity matrix and the zero matrix, respectively. Given a symmetric matrix $A$, the notation $A \succ 0$ indicates that $A$ is positive definite, while $A \succeq 0$ means that $A$ is positive semi-definite. Similarly, $A \prec 0$ indicates that $A$ is negative definite and $A \preceq 0$ means that $A$ is negative semi-definite. For any matrix $A$, $A^{\top}$ denotes the transpose of $A$. The symbol $\mathcal{N}(\mu, \sigma^{2})$ represents a normal distribution with mean $\mu$ and variance $\sigma^{2}$. {Also, the symbol $\mathcal{U}(a, b)$ represents a normal distribution from the interval $[a, b]$. The symbol $\lambda_i$ denotes an eigenvalue of a matrix. 


}

\subsection{Problem Overview}
Consider the discrete time dynamical system
\begin{align}
\label{eqn:system_dynamics}
    x_{k+1} = f(x_k, u_k),
\end{align}
where the state is $x_k \in \mathbb{R}^n$ and $u_k \in \mathbb{R}^m$ is the control input, each at time instant $k \in \mathbb{Z}_{\geq 0}$ with $n, m \in \mathbb{Z}_{>0}$, and ${f}$ is a transition map such that $f: \mathbb{R}^n \times \mathbb{R}^m \mapsto \mathbb{R}^n$, which is generally nonlinear, unknown, and assumed to be stabilizable. 

We consider a scenario in which the system in (\ref{eqn:system_dynamics}) is connected to a controller via a networked medium. Especially, the state readings are provided to the controller through a digital channel, and the controller has direct access to the actuators. The goal is to design a data-driven event-triggered state-feedback controller with gain $K \in \mathbb{R}^{m\times n}$ to stabilize the plant in (\ref{eqn:system_dynamics}) while abiding by a triggering policy that defines the instances $\{{k_i}\}_{{i}\in{\mathcal{Z}}}$ at which a transmission happens, with $\mathcal{Z} \subseteq \mathbb{Z}_{\geq 0}$. At time instant $k = 0$, assume a transmission occurs, so that ${k}_{0} = 0$. In our settings, the controller is updated only upon the violation of some well-defined triggering policy in contrast to the nominal \gls{TTC}.
The sequence $\{{k_i}\}_{{i}\in{\mathcal{Z}}}$ leads to aperiodic updates of the controller. 
The controller then follows a zero-order hold implementation that takes the form of\footnote{{ To be more precise, the control law should be \(u_k = K\xi(x_{k_{i}}), \quad k \in [{k_i, k_{i+1}}),\), where \(\xi(x_{k_{i}})\) is the lifted state. Since we are introducing the vanilla event-triggered mechanism, we use the original state of the system for completeness. In our problem formulation, we adhere to \(u_k = K\xi(x_{k_{i}}),\) unless otherwise stated. The concept of lifting will be discussed in Section \ref{sec:koopman}.}}
\begin{align}
\label{eqn:ET_controller}
    u_k = Kx_{k_{i}}, \quad k \in [{k_i, k_{i+1}}).
\end{align}
The state error takes into account the provided controller's zero-order hold mechanism. 
\begin{align}
    e_k = x_{k_{i}} - x_k,
\end{align}
which can be seen as the deviation between the current state and the last time event {($i$)} is triggered. We consider an event is triggered whenever the following inequality is violated
\begin{align}
\label{eqn:triggereing_policy}
    \|e_k\| \leq \gamma \|x_k\|,
\end{align}
where $\gamma > 0$ is a threshold parameter for the triggering policy. The policy in (\ref{eqn:triggereing_policy}) is evaluated at every time instant $k$, and the control is updated only when the policy is violated. { Fig. \ref{fig:etc} provides an overview of the \gls{ETC} framework. Here the plant $\mathcal{P}$ represents (\ref{eqn:system_dynamics}), the controller $\mathcal{C}$ represents the control law (\ref{eqn:ET_controller}), and the event-triggering policy corresponds to (\ref{eqn:triggereing_policy}).}
\begin{figure}
    \centering
    \includegraphics[width=\linewidth]{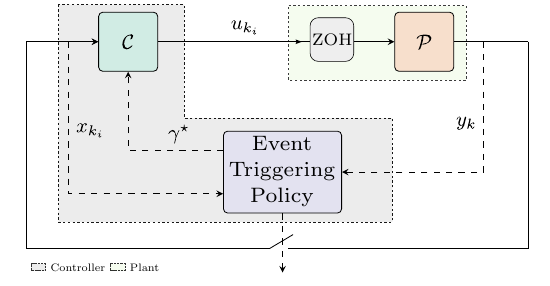}
    \caption{Block diagram visually providing representation that illustrates the core concept underlying \gls{ETC}. It showcases various components and their interconnections, highlighting the essential principles and operational dynamics of the \gls{ETC} framework.}
    \label{fig:etc}
\end{figure}

 \subsection{Persistence of Excitation}
Consider a carried out experiment for the system in (\ref{eqn:system_dynamics}) and its states and input data are recorded in the following way
\begin{align*}
    \mathcal{D} := \{x_k, u_k: k \in [0, (T-1)]\cap\mathbb{Z}_{\geq 0}\},
\end{align*}
where $T$ is the final time of the experiment.
Assume that \(\mathcal{D}\) the dataset exists. Then, we define
\begin{subequations}
\begin{align}
    &U_0:= \mat{u_0~ u_1~ \dots ~ u_{T-1}} \in \mathbb{R}^{m \times T},\\
    &X_0:= \mat{x_0~ x_1~ \dots~ x_{T-1}} \in \mathbb{R}^{n \times T},\\
    & X_1 := \mat{x_1~ x_2~ \dots~ x_T} \in \mathbb{R}^{n \times T}.
\end{align}
\end{subequations}
\begin{assumption} \label{ass:PE}
    Assume $T \geq n+m $, the matrix $\begin{bsmallmatrix}X\\ U_0 \end{bsmallmatrix}$ has full row rank.\hfill  $\Box$
\end{assumption}
Assumption \ref{ass:PE} can be verified numerically for a given set $\mathcal{D}$. The results of Willems et al. \cite{willems2005note} ensures, for discrete-time systems, the validity of assumption \ref{ass:PE} as long as $u$ is a persistently exciting signal.

\section{Framework}
\label{sec:framework}
\subsection{Koopman operator theory}
\label{sec:koopman}

\begin{definition} [Koopman Operator (KO)] Consider the system given in (\ref{eqn:system_dynamics}). The \gls{KO} $\mathcal{K}_t$ is an infinite-dimensional operator \begin{equation}\label{eqn:koopman}
\mathcal{K}_t \xi(x_k) = \xi \circ f(x_k),
\end{equation}
which acts on $\xi \in \mathcal{H}$, where $\mathcal{H}$ is the space of observable functions $\xi: \mathbb{R}^n \to \mathbb{R}$ over the state space, where $\circ$ is the function composition. \hfill $\Box$
\end{definition} 
The \gls{KO} acts on the Hilbert space $\mathcal{H}$ of all scalar measurement functions $\xi$ and is, by definition, a \textit{linear operator}, that is for any $\xi_1, \xi_2 \in \mathcal{H}$ and $\beta_1, \beta_2 \in \mathbb{R}$, we have
\begin{equation}
    \begin{aligned}
    \mathcal{K}_t(\beta_1 \xi_1, \beta_2 \xi_2) &= \beta_1 \xi_1 \circ f + \beta_2 \xi_2 \circ f\\
    &= \beta_1 \mathcal{K}_t \xi_1 + \beta_2 \mathcal{K}_t \xi_2,
\end{aligned}
\end{equation}

An infinite-dimensional space \(\mathcal{H}\) of observable functions is used to represent a nonlinear system linearly using \gls{KO} method \cite{budivsic2012applied}. This means that the dynamics are transformed from nonlinear and finite-dimensional to linear and infinite-dimensional when transitioning from the state-space model to the Koopman representation (see Fig. \ref{fig:Koopman}). However, we are interested in a finite-dimensional approximation of \gls{KO} from a practical perspective. Several approximation methods are addressed in \cite{bevanda2021koopman, proctor2018generalizing}.
\begin{figure}
    \centering
    \includegraphics[width = \linewidth]{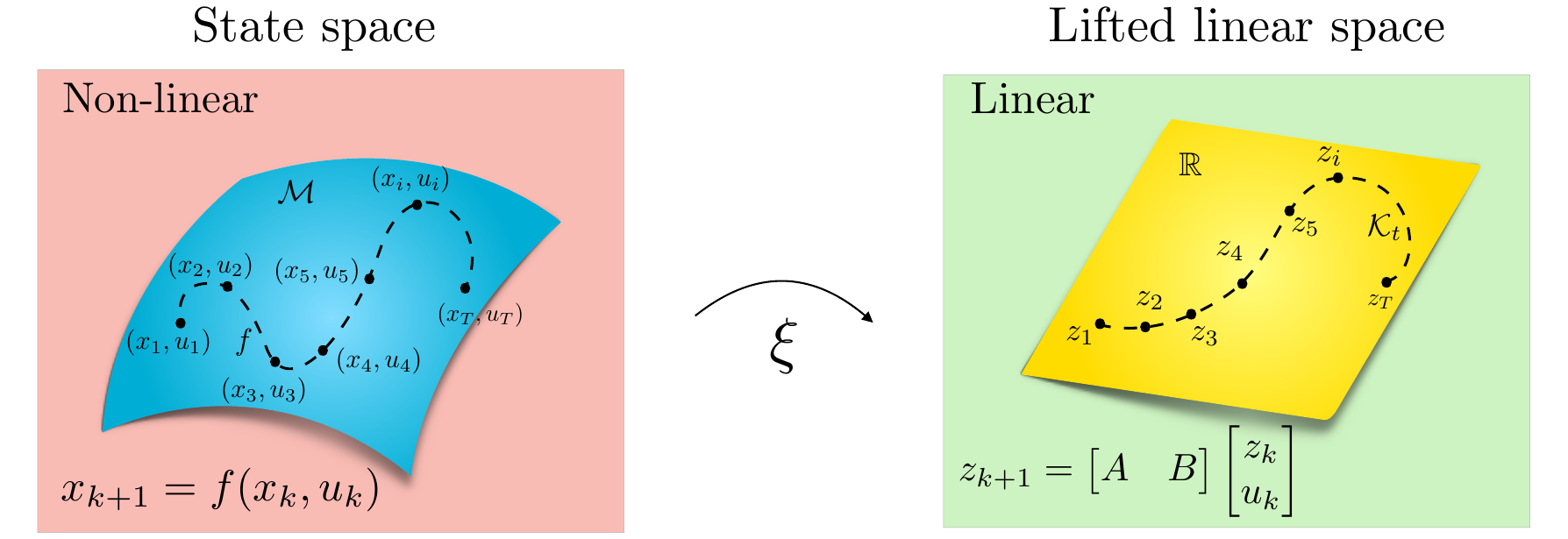}
    \caption{Illustration of the Koopman Operator: The {\color{red}{red}} panel represents the generic nonlinear state-space. Conversely, the {\color{green}{green}} panel represents the linear space.  
}
    \label{fig:Koopman}
\end{figure}

{
To extend this analysis to controlled systems, there exist several methods including \cite{peitz2020data, korda2018linear}. In \cite{peitz2020data}, the authors treated the controlled system as uncontrolled while treating the input as a system parameter. On the other hand, Korda and Mezi\'c \cite{korda2018linear}
dealt with the controlled system in an extended state-space to account for control.

Here, we briefly revisit the approach of \cite{korda2018linear}. In particular, consider the system in (\ref{eqn:system_dynamics}). Let \(\ell(\mathfrak{U})\) be the space of all infinite vectors \(u^{\circ} = \{u_k\}_{k=0}^{\infty}\) with the symbol \(u_{\circ} \in \mathfrak{U}\) and \(\mathfrak{U}\) being an input space. We denote the left shift operator by \(\mathcal{G}^{\star}\) (e.g., \(\mathcal{G}^{\star} u^{\circ}_k = u^{\circ}_{k+1}\)). Also, define \(\mathcal{X}\) to be an extended state such that, \( \mathcal{X} = \begin{bmatrix}
    x_k & u^{\circ}_k \end{bmatrix}^{\top} \). So, the system in (\ref{eqn:system_dynamics}) can be reformulated as,
    \begin{equation}
        \mathcal{X}_{k+1} = \tilde{f} \big(  \mathcal{X} \big) = \mat{f(x_k, u^{\circ}_0) \\ \mathcal{G}^{\star} u^{\circ}_k}.
    \end{equation}
If \(\tilde{\xi} \in \mathcal{H}: \mathbb{R}^n \times \mathbb{R}^m \to \mathbb{R} \) be a new version of the predefined observable function, the Koopman operator \(\mathcal{K}_t: \mathcal{H} \to \mathcal{H}\) for the controlled system turns out to be,
\begin{eqnarray}
    \mathcal{K}_t \tilde{\xi}(\mathcal{X}) = \tilde{\xi} \circ \tilde{f}(\mathcal{X}).
\end{eqnarray}
This was a demonstration of the extension from the uncontrolled systems to the controlled systems. From now on, we will use \(f,\text{and } \xi\) interchangeably between controlled and uncontrolled systems unless otherwise stated.
}

Also, \gls{KO} provides (approximately, in a practical settings) global linear representation for nonlinear dynamics if the right set of observable functions is chosen in as shown in the following. Generally speaking, the observable functions are hard to identify. {They can be found by many method including, but not limited to, brute-force trial and error in a specific basis for the Hilbert space (e.g., trying numerous polynomial functions or Fourier basis functions) or by prior knowledge about the system. Several efforts have been made on this matter \cite{otto2019linearly, yeung2019learning, lusch2018deep, netto2021analytical, kamb2020time, mamakoukas2019local, mamakoukas2021derivative} -- among others.} {Our work relies heavily on the choice of the observable functions. Existing literature on how to choose such dictionary of function can be utilized in order to make best use of the presented paradigm. We discuss the effect of choosing the bad lifting dictionary in illustrative example 2 in section \ref{sec:ex2}}

Motivated by the preceding analysis, we employ the idea of lifting the nonlinear dynamics from its state-space to look linear in a higher-dimensional state-space. 
\begin{remark}
\label{rem:koopman}
    Since our method rely on the good choice of the observable functions, a blurry prior physical knowledge of the underlying plant, not necessarily a complete knowledge, but at least a knowledge that can describe the domain shape in which the system operates to design the observable functions could be of great benefit.  \hfill $\Box$
\end{remark}
\begin{remark}
    At this stage of the work, we design the controller directly from the data. This step requires a set of observable functions that are satisfactory to approximate \gls{KO} as discussed in remark \ref{rem:koopman}. In the sense that we do not focus on the identification of the \gls{KO} itself, we did not include discussion on such a topic. However, in more general scenarios, one may need to identify the operator for any purpose. Readers can refer to \cite{schmid2010dynamic, proctor2016dynamic}. \hfill $\Box$
\end{remark}
Now, the collected set $\mathcal{D}$ should be revised. Instead of having the system's states only, we must consider the additional observable functions taking the form of 
${\Xi}({x}) = \begin{bmatrix}
    \xi_1({x}) & \xi_2({x}) & \dots & \xi_p({x})
\end{bmatrix}^{\top}.$ The observable functions ${\Xi} \in \mathbb{R}^{p} (p > n)$. Note that we only lift the state not the control. So the set $\mathcal{D}$ becomes
\begin{subequations}
\label{eqn:lift_Xs}
\begin{align}
    &U_0:= \mat{u_0~ u_1~ \dots ~ u_{T-1}} \in \mathbb{R}^{m \times T},\\
    &Z_0 := \Xi(X_0) \in \mathbb{R}^{p \times T},\\
    &Z_1 := \Xi(X_1) \in \mathbb{R}^{p \times T}.
\end{align}
\end{subequations}
\begin{remark}
    In response to this change, a slight modification of assumption \ref{ass:PE}, replacing the previous condition \(T \geq n + m\) with \(T \geq p + m\). \hfill $\Box$
\end{remark}
In response to this, the condition in (\ref{eqn:triggereing_policy}) becomes,
\begin{align}
\label{eqn:triggering_policy_lifted_explicit}
    \|e_k^{\xi}\| \leq \gamma \|\xi(x_k)\|,
\end{align}
where \(e_k^{\xi} := \xi(x_k) - \xi(x_{k_i})\).
Hence, after choosing the set of observable functions, the system in (\ref{eqn:system_dynamics}) can be now formulated as
\begin{subequations}
\label{eqn:lifted_system}
\begin{align}
    z_{k+1} & = Az_k + Bu_k,\\
    x_k & = Cz_k.
\end{align}
\end{subequations}

\subsection{Event-triggered Control for the Lifted Representation of the Non-linear Dynamics}

Consider the system given in (\ref{eqn:lifted_system}), the globally linear version of the system in (\ref{eqn:system_dynamics}), subject to the controller (\ref{eqn:ET_controller}) that results in
\begin{subequations}
\label{eqn:closed_loop}
\begin{align}
    z_{k+1} &= Az_k + BKz_{k_i}\\
    &= Az_k + BKz_k + BKz_{k_{i}} - BKz_k\\ 
    &= (A+BK)z_k + BKe_k, \quad \forall k \in [{k_i, ~k_{i+1}}),
\end{align}
\end{subequations}
which can be understood as a closed-loop representation of the system in (\ref{eqn:lifted_system}) with the state error.

An alternative representation of the event-triggered closed loop system should be derived to account for the data-driven nature of this work. In ref. \cite{de2023event} the authors derived a data-driven representation of the closed loop system without considering the matrix $BK$. On the other hand, Digge and Pasumarthy \cite{digge2022data} developed a closed loop representation that allows dealing with the event-triggered formulation. The representation in \cite{da2018data} is modified to account for the lifted linear representation of the nonlinear dynamics.
\begin{lemma}[Data-driven representation \cite{digge2022data, de2019formulas}]
The equivalent data-driven closed loop representation of the system (\ref{eqn:closed_loop}) under satisfaction of assumption \ref{ass:PE} and where     \begin{align} 
        \mat{I \\ K} = \mat{Z_0\\ U_0} L, ~ \text{and} ~
        \mat{\underbar{0} \\ K} =  \mat{Z_0\\ U_0} N,
        \label{eqn:alt_rep2}
    \end{align} holds, takes the following form
    \begin{align}
    \label{eqn:alt_rep}
        z_{k+1} = Z_1L z_k + Z_1 N e_k,
    \end{align}
    where $L$ and $N$ are $T\times p$ matrices. \hfill $\Box$
\end{lemma}
\begin{proof}
Let assumption \ref{ass:PE} be satisfied. Hence, by the Rouch\'e-Capelli theorem, there exist matrices \(L\) and \(N\) that satisfy (\ref{eqn:alt_rep2}). So, another representation of (\ref{eqn:alt_rep}) can be written as
\[
z_{k+1} = \begin{bmatrix} A & B \end{bmatrix} 
\begin{bmatrix} I \\ K \end{bmatrix} 
z_k + 
\begin{bmatrix} A & B \end{bmatrix} 
\begin{bmatrix} \underbar{0} \\ K \end{bmatrix} 
e_k.
\]
Using (\ref{eqn:alt_rep2}), the closed-loop system is given by
\[
z_{k+1} = 
\underbrace{
\begin{bmatrix} A & B \end{bmatrix}
\begin{bmatrix} Z_0 \\ U_0 \end{bmatrix}
}_{Z_1}
L z_k + 
\underbrace{
\begin{bmatrix} A & B \end{bmatrix}
\begin{bmatrix} Z_0 \\ U_0 \end{bmatrix}
}_{Z_1} 
N e_k.
\]
Therefore, the data-driven representation of the closed-loop system (\ref{eqn:closed_loop}) obtained as stated in  (\ref{eqn:alt_rep2}).
\end{proof}
This formulation can be considered as a reparametrization of the system in (\ref{eqn:closed_loop}) in terms of data. In other words, no need for the prior explicit system identification step. Having established this formulation, we now proceed to derive the condition for system (\ref{eqn:alt_rep}) to be exponentially stable in Lyapunov sense. 
A linear system described by \(z_{k+1} = Az_k\), where \(A \in \mathbb{R}^p \), is considered exponentially stable if there exists a function \(V:\mathbb{R}^p \to \mathbb{R}\) defined by 
\( V(z_k) = z_k^\top S z_k \)
with \( S \succ 0 \) and symmetric, such that 
\( V(z_{k+1}) \leq \alpha V(z_k) \)
along the system's trajectories for all \( k \geq 0 \) and for some \(\alpha \in \mathcal{A} := (0, 1] \subset \mathbb{R}\).
\begin{remark}
\label{rem:alpha}
    For unstable systems, the choice of $\alpha$ is critical as it impacts the values of the controller gain $K(\alpha)$ which must satisfy the necessary conditions and thresholds to stabilize the system. {This condition can be formalized as, \[
    \alpha^{\star} = \inf_{\alpha \in \mathcal{A}} \big\{ \alpha : K(\alpha) \implies |\lambda_i| < 1, ~\forall \lambda_i \big\},
\] where the \(K(\alpha)\) is the gains corresponding to one value of $\alpha$ on \(\mathcal{A}\).} \hfill $\Box$
\end{remark}
Consider the classical Lyapunov candidate function described below, the exponential Lyaponuv stability criteria\footnote{The full analysis is given in the appendix.} is given by
\begin{align}
\label{eqn:Lyapunov}
 \begin{bmatrix} z_k \\ e_k \end{bmatrix}^{\top}  \begin{bmatrix}
L^\top Z_1^\top S Z_1 L - \alpha S & L^\top Z_1^\top S Z_1 N \\ 
N^\top Z_1^\top S Z_1 L & N^\top Z_1^\top S Z_1 N 
\end{bmatrix} \begin{bmatrix} z_k \\ e_k \end{bmatrix} \leq 0
\end{align}

In this work, the design of the \gls{ETC} strategy should not violate the Lyapunov stability condition in (\ref{eqn:Lyapunov}) to ensure exponential stability.

\subsection{Learning Controller From Data}
Firstly, we design the controller gains to stabilize the globally linearized system. We consider the data-driven closed loop representation in (\ref{eqn:alt_rep}) neglecting the error at this stage
\begin{align}
\label{eqn:no_error_dynamics}
    z_{k+1} = Z_1 L z_k,
\end{align}
the controller gains can be designed directly from data, as discussed in \cite[Section IV.~A]{de2019formulas}. Further, the following theorem ensures the Lyaponuv stability condition.
\begin{theorem}[{Direct Controller Design}]
    Let condition \ref{ass:PE} hold. And by exploiting the results of lemma 1. Then any matrix $G_1$ that satisfy the following \gls{LMI},
    \begin{align}
    \label{eqn:stability_lmi}
        \mat{Z_0 G_1 & G_1^{\top} Z_1^{\top}\\
        Z_1 G_1 & Z_0 G_1} \succeq 0
    \end{align}
    results in
    \begin{align}
    \label{eqn:cont_gain}
        K = U_0 G_1 (Z_0 G_1)^{-1}
    \end{align}
    which stabilizes the system (\ref{eqn:lifted_system}). \hfill $\Box$
\end{theorem}
\begin{proof}
    To check the stability in exponential decay of the system (\ref{eqn:no_error_dynamics}) with a rate $\alpha$, implies
    \begin{align}
    \label{eqn:stability}
        L^{\top}Z_1SZ_1L - \alpha S \preceq 0,
    \end{align}
    with $L$ satisfying (\ref{eqn:alt_rep2}). Let $G_1 := LS^{-1}$, and pre- and post-multiply (\ref{eqn:stability}) by $S^{-1}$, the  stability of the system can be guaranteed if there exists two matrices $G_1$ and $S$ such that
    \[
    \begin{array}{l}
    G_1^{\top} Z_1^{\top} S Z_1 G_1 - \alpha S^{-1}\preceq 0 \\
    K S^{-1} = U_0 G_1 \\
    S^{-1} = Z_0 G_1
    \end{array}
    \]
    Moreover, we use $S^{-1} = Z_0 G_1$ and obtain
    \[
    \begin{array}{l}
    G_1^{\top} Z_1^{\top} (Z_0 G_1) Z_1 G_1 - \alpha Z_0 G_1 \preceq 0 \\
    Z_0 G_1 \succ 0 \\
    K = U_0 G_1(Z_0 G_1)^{-1}
    \end{array}
    \]
    Using Schur's complement lemma on the first inequality, we reach to (\ref{eqn:stability_lmi}) which results in gains given from (\ref{eqn:cont_gain}) that exponentially stabilize the system.
\end{proof}

\subsection{Learning the Triggering Policy from Data}
In the interval \( [{k_i, ~k_{i+1}}) \), it is essential that inequality (\ref{eqn:Lyapunov}), which ensures exponential convergence, is also satisfied. The following theorem derives a window for the parameter \(\gamma\) that ensures the stability of system (\ref{eqn:alt_rep}).

\begin{theorem}[{Optimal Threshold}]
\label{th:gamma}
    Assume that the condition \ref{ass:PE} is satisfied. So, the relative threshold parameter \(\gamma\) for the event-triggered implementation (\ref{eqn:triggereing_policy}) with the controller (\ref{eqn:cont_gain}) can be calculated by solving for \(\gamma\) such that
    \begin{equation}
    \label{eqn:gamma}
        \begin{aligned}
            &\max_{\substack{q, G_2}} \quad \gamma \\
            &\text{s.t.} \\
            &\left[
            \begin{array}{cccc}
            \alpha Z_0 G_1 & \underbar{0} & G_1^\top Z_1^\top & \gamma Z_0 G_1 \\
            \underbar{0} & q I & G_2^\top Z_1^\top & \underbar{0} \\
            Z_1 G_1 & Z_1 G_2 & Z_0 G_1 & \underbar{0} \\
            \gamma Z_0 G_1 & \underbar{0} & \underbar{0} & q I
            \end{array}
            \right] \succeq \underbar{0},\\
            & q > 0, \quad Z_0 G_2 = 0, \quad U G_2 - q K = 0,
        \end{aligned}
    \end{equation}
which will result in stability of the system (\ref{eqn:alt_rep}) in exponential behaviour. \hfill $\Box$
\end{theorem}
{\begin{proof}
For exponential stability during event-triggered control, whenever the triggering condition (\ref{eqn:triggereing_policy}) is met, the condition (\ref{eqn:Lyapunov}), which guarantees stability, must hold as well. This relationship can be encoded using the S-procedure \cite{boyd2004convex}.
According to the S-procedure, (\ref{eqn:triggereing_policy}) implies (\ref{eqn:Lyapunov}) if there exists a constant \(\eta \geq 0\) such that:
\[
\eta
\begin{bmatrix}
-\gamma^2 I & \underbar{0} \\
\underbar{0} & I
\end{bmatrix}
\preceq 
\begin{bmatrix}
L^{\top} Z_1^\top S Z_1 L - \alpha S & L^{\top} Z_1^{\top} S Z_1 L \\
L^{\top} Z_1^{\top} S Z_1 L & L^{\top} Z_1^{\top} S Z_1 L
\end{bmatrix}.
\]
Using Schur's complement, and post- and pre-multiplying by the \(\texttt{diag}(S^{-1}, I, I)\), we derive:
\[
\begin{aligned}
&\begin{bmatrix}
-\eta \gamma^2 S^{-2} + \alpha S^{-1} & \underbar{0} & S^{-1} L^\top Z_1^\top \\
\underbar{0} & \eta I & N^\top Z_1^\top \\
Z_1 L S^{-1} & Z_1 N \eta^{-1} & S^{-1}
\end{bmatrix} \succeq \underbar{0}.
\end{aligned}
\]
By changing the variables \(G_1 = L S^{-1}\), \(G_2 = \eta^{-1} N\), \(q = \eta^{-1}\), and \(S^{-1} = Z_0 G_1\), we arrive at the LMI:
\[
\begin{aligned}
    &\begin{bmatrix}
    \alpha Z_0 G_1 & \underbar{0} & G_1^\top Z_1^\top & \gamma Z_0 G_1 \\
    \underbar{0} & qI & G_2^\top Z_1^\top & \underbar{0} \\
    Z_1 G_1 & Z_1 G_2 & Z_0 G_1 & \underbar{0} \\
    \gamma Z_0 G_1 & \underbar{0} & \underbar{0} & qI
    \end{bmatrix} \succeq \underbar{0}.
\end{aligned}
\]
The result of theorem \ref{th:gamma} allows to maximize $\gamma$ over the variables $G_2$ and $q$. The result also implies that any $\gamma \in [{0, \gamma^{\star}})$ stabilizes the system, where $\gamma^{\star}$ is the solution for (\ref{eqn:gamma}).
\end{proof}
}

Now, we have all the components put together. A detailed algorithm for the entire process is given in algorithm \ref{alg:data_drive_ETC}. 

\begin{algorithm}[htp]
\caption{Koopman Operator-Based Event-Triggered Control}
\label{alg:data_drive_ETC}
\begin{algorithmic}[1]
\Require $\alpha$, $X_0$, $X_1$, and $U_0$
\State  Lift $X_0$, and $X_1$ via (\ref{eqn:lift_Xs} b, and c)
\State Solve for $G_1$ in the \gls{LMI} given in (\ref{eqn:stability_lmi})
\State Solve for the controller gain $K$ in (\ref{eqn:cont_gain})
\State Maximize the threshold parameter $\gamma$ to get $\gamma^{\star}$ in (\ref{eqn:gamma})
\State  Choose any $\gamma \in [{0, \gamma^{\star}}]$, (typically the max. value gives wider inter-event time window)\\
\Return $\gamma^{\star}$, and $K$
\end{algorithmic}
\end{algorithm}
\section{Illustrative Simulations and Results}
{

In this section, three numerical examples are selected to illuminate distinct aspects of the proposed algorithm. Example 1 in Sec.~\ref{sec:ex1} admits exact linearisation through a suitably chosen observable, making it an effective test bed for parameter studies. In this setting, we investigate how the decay rate~$\alpha$ influences performance by tracking the peak discrete derivative of the Lyapunov function, and we assess the algorithm’s ability to stabilise the system from a range of initial conditions. Additionally, example 2 in Sec.~\ref{sec:ex2} addresses a system for which, to the best of the authors’ knowledge, no observable yields a closed‑form exact linearisation. This case isolates the consequences of observable selection and highlights how the proposed method behaves when exact linearisation is unavailable. Finally, example 3 in Sec.~\ref{sec:ex3} demonstrates the method’s generality for linear systems, as formalised in the corresponding section.

Following this section, we provide additional notes regarding some critical aspects of the proposed method.





%
}
\subsection{Illustrative Example 1: Proof of Concept}
\label{sec:ex1}
We consider a case of nonlinear system with slow manifold used in relative works \cite{brunton2016koopman, surana2016linear, surana2016koopman}:
\begin{align}
\label{eqn:nonlinear_ex}
\begin{bmatrix}
x_1 \\
x_2
\end{bmatrix}
\mapsto
\begin{bmatrix}
\rho x_1 \\
\kappa x_2 + (\rho^2 - \kappa) x_1^2 + u
\end{bmatrix}.
\end{align}

In this scenario, there exists a polynomial stable manifold defined as \(x_2 = x_1^2\). 
Within the Koopman-inspired framework, if the correct observable functions were chosen such that $\Xi(x) = \mat{x_1 & x_2 & x_1^2}^{\top}$, the nonlinear system in (\ref{eqn:nonlinear_ex}) can be expressed linearly as
\begin{align} \label{eqn:koopman_sys}
\begin{bmatrix}
z_1 \\
z_2 \\
z_3
\end{bmatrix}_{k+1}
=
\begin{bmatrix}
\rho & 0 & 0 \\
0 & \kappa & (\rho^2 - \kappa) \\
0 & 0 & \rho^2
\end{bmatrix}
\begin{bmatrix}
z_1 \\
z_2 \\
z_3
\end{bmatrix}_k + \mat{0\\1\\0} u_k
\end{align}
 Considering the parameters for the system, $\rho = 0.6$, and $\kappa = 1.2$, the corresponding eigenvalues are $\lambda_1 = 0.6$, $\lambda_2 = 1.2$, and $\lambda_3 = 0.36$. Since $\lambda_2 > 1$, the system exhibits instability and the goal is to stabilize the trajectory around the origin. We collected the data for $T = 45$ which is enough for assumption \ref{ass:PE} to hold -- on a theoretical note, $T \geq m + p$ samples should be enough (i.e. in this example $T \geq 4$) to obey assumption \ref{ass:PE}. Therefore, $T = 4$ should work. The input signal is drawn from a normal distribution following $u \sim \mathcal{N}(0, 1)$.

Then, after deploying the steps in algorithm \ref{alg:data_drive_ETC}, we obtain $K = \mat{0.0206  & -1.1109 &  -0.1530}$, which in turn gives $\gamma^{\star} = 0.7664$. We simulated the system for both \gls{ETC}, and \gls{TTC} and illustrated the behavior in Fig. \ref{fig:results}. All the results depicted in Fig. \ref{fig:results}, are acquired after pulling the states back from the higher-dimensional space, in this case from $\mathbb{R}^3$ to $\mathbb{R}^2$, by applying (\ref{eqn:lifted_system}b) with \(C = \begin{bsmallmatrix}
    I_2 & 0\\ 0 & 0 
\end{bsmallmatrix}\).

\begin{figure*}[]
    \centering
    \includegraphics[width = \textwidth]{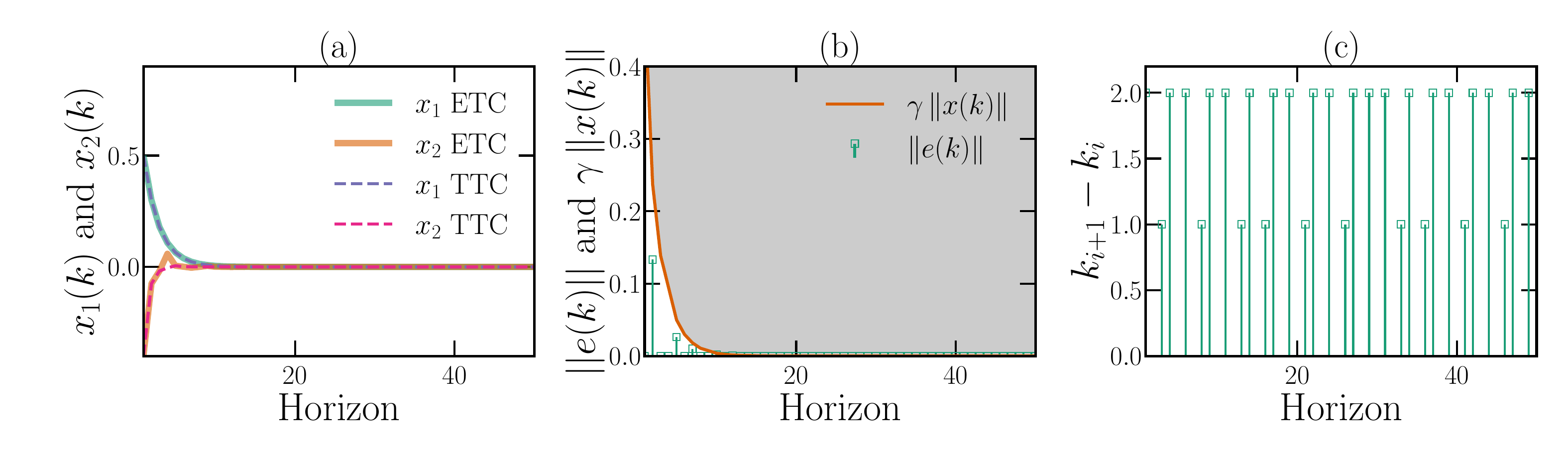}
    \caption{Results of the illustrative example 1. (a) Behaviour of state \(x_1\) and \(x_2\) for both ETC and TTC over the horizon. (b) Norms of the error \(\|e_k\|\) and the threshold parameter \(\gamma \|x_k\|\). (c) Inter-event times \(k_{i+1} - k_i\) showing the intervals between successive events.}
    \label{fig:results}
\end{figure*}

Fig. \ref{fig:results}(a), illustrates the state evolution \(x_1\) and \(x_2\) against time under both \gls{ETC} and \gls{TTC} techniques. The trajectories for \gls{ETC} demonstrate excellent tracking performance in comparison with the nominal \gls{TTC}. This highlights the efficacy of the developed event-triggered approach in maintaining system stability while minimizing unnecessary updates. 

Also, in Fig. \ref{fig:results}(b), the graph shows that \(\|e_k\|\) remains consistently below \(\gamma \|x_k\|\), satisfying the triggering condition. As shown in Fig. \ref{fig:results}(c), the substantial reduction in communication instances ($40\%$) addresses potential concerns regarding communication overhead in practical implementations.


{Finally, as noted from the numerical results, ETC not only achieves comparable performance to TTC but does so with fewer communication instances (about 40 compared to 100) and lower control cost (approximately 0.203 vs. 0.210), which supports our hypotheses. Another note in our experiment, both Koopman based linearization \gls{ETC} and \gls{TTC} have control cost much lower than the traditional Taylor linearization technique, consistent with the results of Brunton et al. \cite{brunton2016koopman}.}

For this example, we want to thoroughly examine ammd understand the influence of various parameters on the system's behavior and stability. This includes analysis of how different initial conditions and the parameter \(\alpha\) impact the system dynamics. We explore these effects through a series of extensive simulations, designed to provide a comprehensive view of the system’s response under a range of scenarios to enrich our theoretical insights and understanding. 

Initially, we assessed the robustness of the algorithm by simulating ten different random initial conditions drawn from a uniform distribution \(\sim \mathcal{U}(-5, 5)\). Fig. \ref{fig:rnd_init_conds} shows the behavior of both \(x_1\) and \(x_2\) while starting with those random initial conditions. The figures show that while the initial conditions varies significantly, the behaviour of the system states stabilizes in a finite amount of time. An interesting observation from the same figure is that the error decay rate between the state and the reference in the log scale is nearly linear, supporting the paper's earlier demonstration of the exponential error decaying property.
\begin{figure}[]
    \centering
    \includegraphics[width=\linewidth]{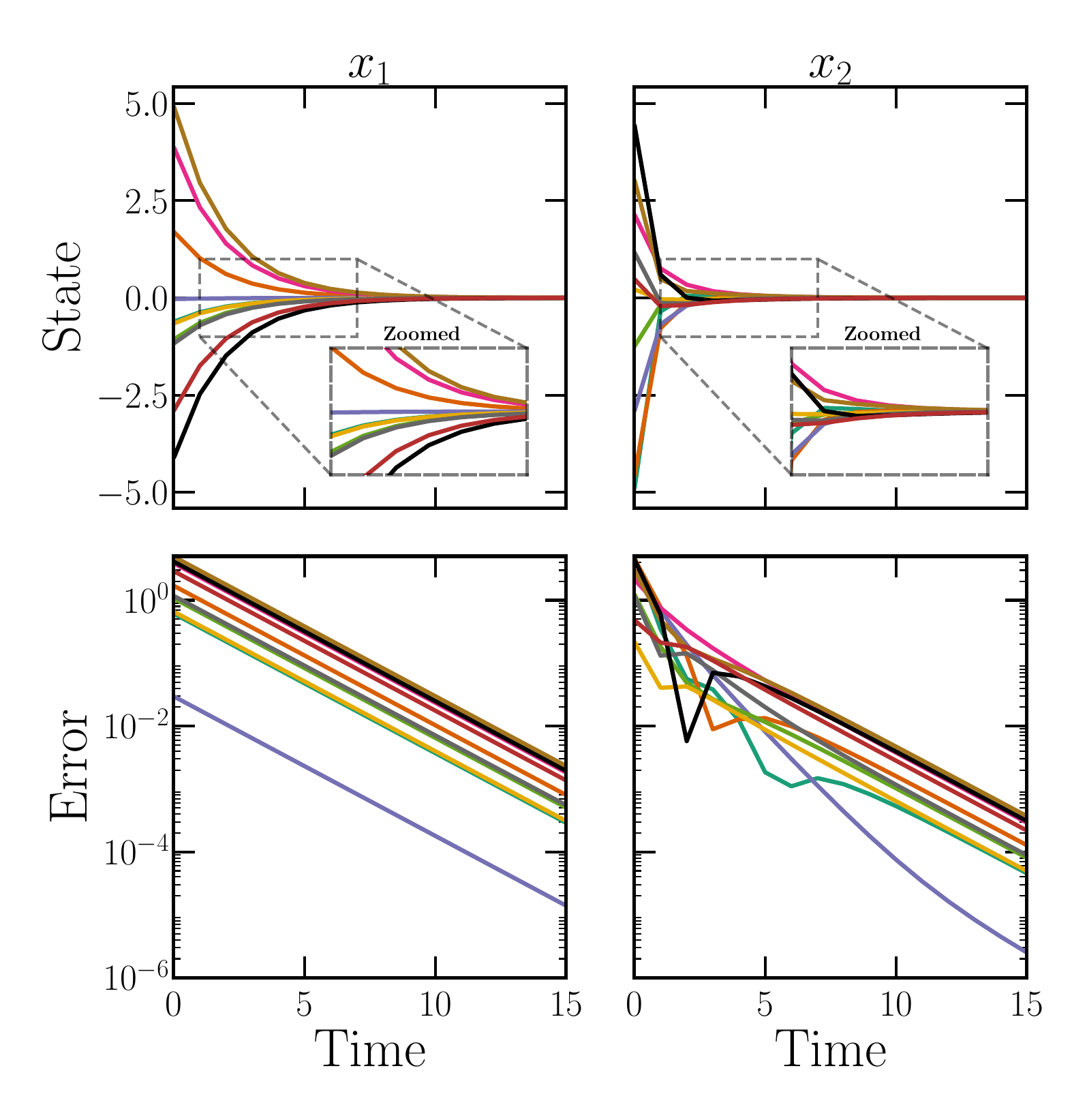}
    \caption{{A simulation of ten random initial conditions drawn from a uniform distribution $X \sim \mathcal{U}(-5, 5)$. The figure shows the behaviour of \(x_1\) (left), and $x_2$ (right).}}
    \label{fig:rnd_init_conds}
\end{figure}

Subsequently, the initial conditions were fixed at $x_0 = \begin{bmatrix}
    0.5 & -0.4
\end{bmatrix}$ simulations were conducted across a fine grid of different \(\alpha\) values ranging from 0.4 to 1. The choice of 0.4 as the starting value is informed by empirical observations, which indicate that this value represents the minimum threshold necessary to achieve an adequate gain for system stabilization, as detailed in remark \ref{rem:alpha}. Fig. \ref{fig:alpha} demonstrates that for each value of $\alpha$, there is no violation in the rate of Lyapunov function decay. The values on the x-axis in this figure must not exceed their corresponding values on the y-axis (i.e. they cannot cross the line \(\max \big(V(k+1)/V(k)\big) = \alpha \)). In other words, no deviation from the expected decaying behavior is observed.

\begin{figure}[]
    \centering
    \includegraphics[width=\linewidth]{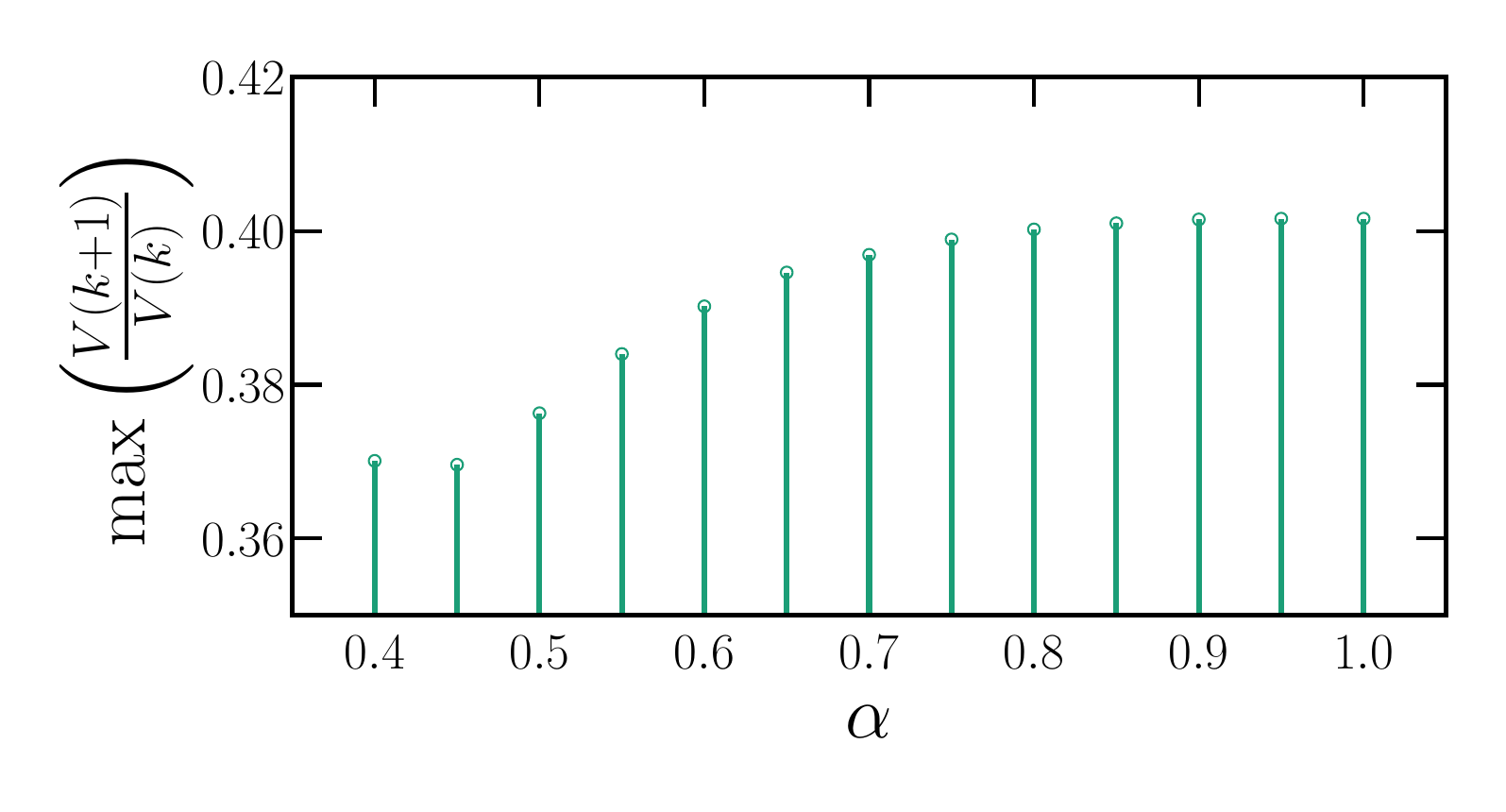}
    \caption{{
    The relationship between \(\alpha\) and the Lyapunov function decay rate. Simulations confirm no violations in the decay rate, as all points lie below the boundary \(\max \big(V(k+1)/V(k)\big) = \alpha \), ensuring system stability across the tested \(\alpha \) range.}}
    \label{fig:alpha}
\end{figure}

\begin{figure*}
    \centering
    \includegraphics[width=\textwidth]{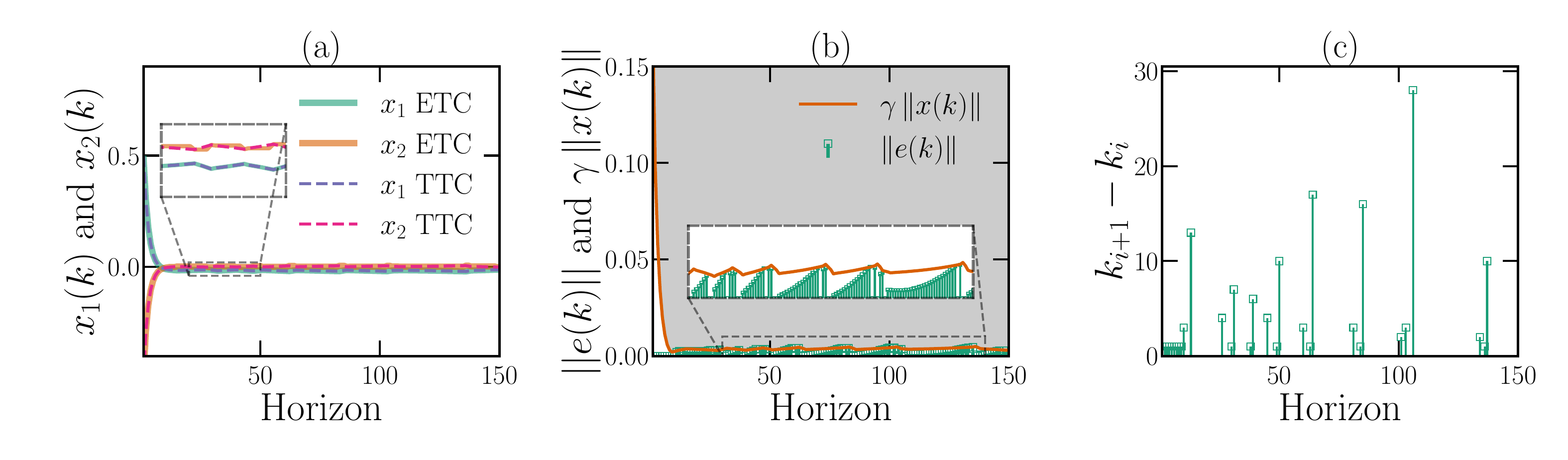}
    \caption{Results of the illustrative example 2. (a) Behaviour of state \(x_1\) and \(x_2\) for both ETC and TTC over the horizon. (b) Norms of the error \(\|e_k\|\) and the threshold parameter \(\gamma \|x_k\|\). (c) Inter-event times \(k_{i+1} - k_i\) showing the intervals between successive events.}
    \label{fig:poly_event_results}
\end{figure*}
In terms of time analysis, using a persistently excited signal sequence of length 45, the controller learning process took 0.4 seconds on an M2 MacBook Air, demonstrating the algorithm’s practical efficiency. Similarly, learning the triggering policy under the same setup was achieved in 0.49 seconds.

\subsection{Illustrative Example 2: Polynomial System}

{
In this example, we consider a discrete-time nonlinear polynomial system defined as:
\begin{align}
\label{eq:poly_sys}
\begin{bmatrix}
x_1 \\
x_2
\end{bmatrix}
\mapsto
\begin{bmatrix}
a x_1 + b x_1^2 + c x_2 \\
d x_1 x_2 + e x_2^2 + u
\end{bmatrix},
\end{align}
 with parameters \( a = 1.05 \), \( b = 0.1 \), \( c = 0.5 \), \( d = 0.25 \), and \( e = 0.08 \). For this system, data were collected for \(T = 150\) which is sufficient for assumption 1 to hold. In this example, we gather the data in a closed-loop manner inspired by Ref. \cite{do2024practical} practical guidelines. In their work, they gather the identification data using a controller architecture similar to the one intended for design. For example, in the current example, we wish to design a stabilizing controller in the form of \(u = Kz\). Therefore, we adopt a closed loop controller of a similar structure, \(u = \tilde{K}z\), where the over all signal is perturbed by a random signal to ensure excitation. Perturbed trajectories are generated starting from a nonzero initial condition. The choice of the observable functions are motivated by the nonlinear terms in the system plus higher order polynomials as follows: 
\(
\Xi(x) = 
\begin{bmatrix}
1 & x_1 & x_2 & x_1^2 & x_1 x_2 & x_2^2 & x_1^3 & x_2^3 &  x_1^3 & x_2^3 
\end{bmatrix}^\top.
\) The addition of the higher polynomials in the dictionary of the observable functions helped to reduce the steady state error as shown in Fig. \ref{fig:poly_sys_degree}. With the previous parameters and the choosing observables, in addition to choosing \(\alpha = 0.7\), we proceed to apply Algorithm \ref{alg:data_drive_ETC}. As a result, we got \(K = \mat{-0.8742  & -0.3989  & -0.3779   & -0.5845  & -0.1581 \dots \\
\dots  0.0301 &   0.0463  & -0.0157    & 0.0084  & -0.0071}\), and \(\gamma = 0.2237\). The corresponding results and the behavior of the system are shown in Fig. \ref{fig:poly_event_results}. Although the interpretations of these results are quiet similar to that of Section \ref{sec:ex1}, we still need to asses the choice of the observable functions; since the system does not obey a choice of observables that exactly linearize the nonlinear system in a closed from as done in Section \ref{sec:ex1}. This is known as the closure problem.  

To asses the effect of observable functions, we tested multiple cases under different observables. We adhere to monomial lifting functions with increasing order. Firstly, one can see the weak representation of a first-order polynomial for observable functions. This means that we implicitly assume that the nonlinear system in (\ref{eq:poly_sys}) can be well represented using linear observables, which is not evident from the most left bar of Fig. \ref{fig:poly_sys_degree} following our understanding. Furthermore, a clear decreasing pattern appears for the steady state error as the degree number increases from 2 to 5. Surprisingly, the steady state error starts to increase at degree 6 which is a sign of overfitting. Consequently, increasing the number of observables does not guarantee a good representation. Further discussion on this will be provided in Section \ref{sec:disccussions_obs}.
\label{sec:ex2}}

\begin{figure}
    \centering
    \includegraphics[width=\linewidth]{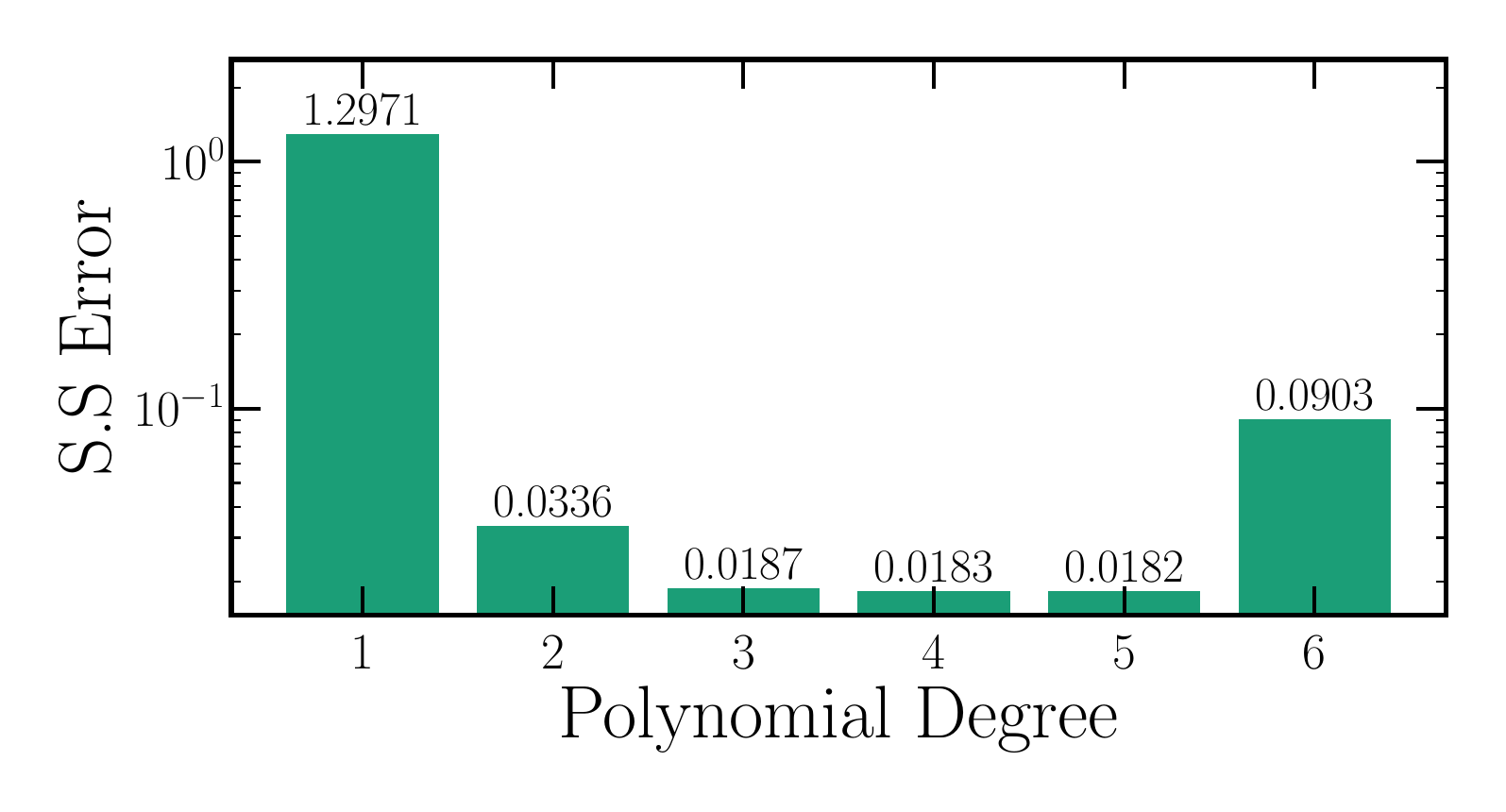}
    \caption{Steady-state error versus polynomial degree of the lifting function used for controller synthesis.}
    \label{fig:poly_sys_degree}
\end{figure}


\subsection{Illustrative Example 3: Linear Case Note}
{
\label{sec:ex3}
In this example, we consider the discrete-time linear system (originally presented in \cite{digge2022data})
\[
x_{k+1} = Ax_k + Bu_k,
\]
with \(x_k \in \mathbb{R}^2\) and \(u_k\) generated from a random control signal within the interval \([-3,3]\). The system matrices are given by
\(
A = \left[ \begin{smallmatrix} 0.91 & 0.0995 \\ 0.02 & 0.99 \end{smallmatrix}\right], \quad B = \left[ \begin{smallmatrix} 0.01 \\ -0.184 \end{smallmatrix}\right],
\)
and the simulation is performed over 20 data points to learn a controller gain \(K\) with \(\alpha = 0.9\) as well as the triggering threshold \(\gamma\). Under the choice of the observable functions as identity, the gain \(K\) and an event-triggering threshold \(\gamma\) are computed according to the framework in \ref{alg:data_drive_ETC} and resulted in \(K = \mat{0.3820 & 0.8179}\) and \(\gamma = 0.4653\).

Notably, if the observable functions are chosen as the identity, then the algorithm directly addresses the linear dynamics without any additional lifting. In this case, our approach reduces to the classical linear control and event-triggered control framework as presented in \cite{digge2022data}. The results of this example are depicted in Fig. \ref{fig:linear_sys} which matches with excellent agreement the results in Digge and Pasumarthy \cite{digge2022data} work.

\begin{figure*}
    \centering
    \includegraphics[width=\linewidth]{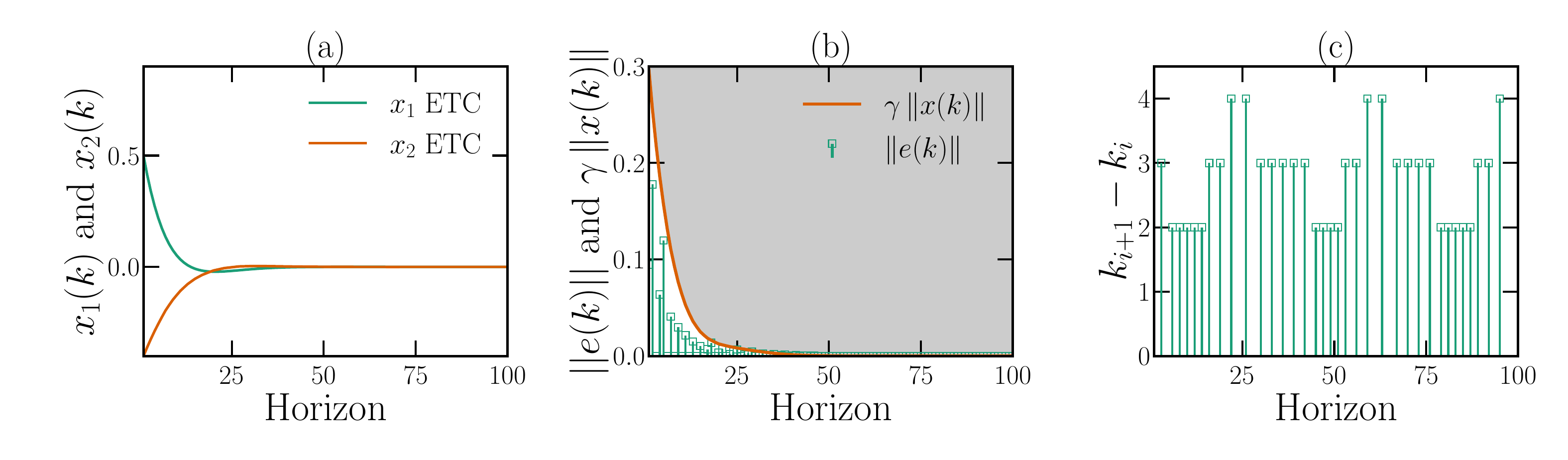}
    \caption{Results of the illustrative example 3. (a) Behaviour of state \(x_1\) and \(x_2\). (b) Norms of the error \(\|e_k\|\) and the threshold parameter \(\gamma \|x_k\|\). (c) Inter-event times \(k_{i+1} - k_i\) showing the intervals between successive events.}
    \label{fig:linear_sys}
\end{figure*}

}

\section{Discussions}
\subsection{On the choice of the observable functions}
\label{sec:disccussions_obs}
{
We devote this section to discussing some notes towards practical discussion of the selection of observable functions. As mentioned earlier in the paper, the choice of the lifting functions is critical part in designing the direct data driven controller that abides to an additional event-triggering role. This is also noted in the literature covering the Koopman operator \cite{williams2015data}. The authors noted that the rate of convergence of the Koopman based identification method depends on the dictionary of the observable functions which they refer to as trials or basis functions. In their work they assume the basis functions spans the subspace of the observables. One can motivate the choice of the observables in our methods using similar approach. Possible choices of the observables are: polynomials, and \glspl{RBF}. As advised by the authors, polynomials are a good option to be used for the system defined on \(\mathbb{R}^N\), and the \glspl{RBF} are a good candidate for systems defined on irregular domains.

Recently, considerable research effort has been devoted to the idea of finding a good set of observables. Namely: 
\begin{enumerate}[label=\roman*)]
    \item employing deep learning autoencoders \cite{lusch2018deep, ge2024deep},
    \item local higher-order derivatives of the nonlinear system \cite{mamakoukas2019local} (though it requires a symbolic expression of the nonlinear system),
    \item motivating the choice of observables by analytical constructions in \cite{netto2021analytical} and \cite{shi2021acd}. In \cite{shi2021acd} the method is specifically tailored to certain robotic systems by exploiting the topological spaces of these systems to guide the construction of Hermite polynomial-based observables,
    \item avoiding explicit basis selection by defining observables implicitly in a reproducing kernel Hilbert space \cite{williams2014kernel, lee2024kernel} using kernel functions (e.g. Gaussian kernel \(k(x, y) = e^{\|x-y\|^2/2\sigma^2}\)).
\end{enumerate}

Building on the above discussion, how to choose an appropriate set of observables remains an important, yet open, question. Nonetheless, many of the methods discussed above have shown practical effectiveness.
}


\subsection{On the Controller Synthesis}
{
In this section we refer back to the examples in Section \ref{sec:ex1}, and Section \ref{sec:ex2}. Although the feedback control law looks linear on the observable space \(z\), it could be looked to it as a nonlinear controller on the original state \(x\). For example, the illustrative example 1 in Section \ref{sec:ex1}, the designed controller in terms of the original states can be interpreted as,
\begin{align}
\begin{aligned}
    u &= \mat{K_1 & K_2} \mat{x_1 \\ x_2} + K_3 x_1^2 \\
    & = 0.0206\;x_1 -1.1109\; x_2 - 0.1540\;x_1^2
\end{aligned}
\end{align}
where the third term indicates the nonlinearity imposed by the choice of the nonlinear observable function \(x_2^2\). Similar analysis could be done for example 2 by examining the effect of the observables on the control law. On the other hand, in example 3 in Section \ref{sec:ex3}, the control law remains linear as it originally introduced.

We recommend that the future research should look into how to optimally design the controller gain and also a policy that respects the system stability. 
}

\section{Conclusion}
To sum up, this study proposes an event-triggered control approach based on data-driven methods for discrete-time nonlinear systems. By lifting the nonlinear dynamics into a higher-dimensional linear representation inspired by the \gls{KO} theory, the method makes it possible to create an event-triggered controller driven by data. Through the development of a closed-loop system and the implementation of a triggering strategy, the proposed method stabilizes the plant with less frequent control updates. 

The event-triggered closed-loop system's exponential stability is guaranteed by the stability analysis, which is based on the Lyapunov criterion. Numerical simulations and theoretical analysis are used to show how effective the suggested strategy is. This work creates opportunities for real-world applications in networked control systems and advances event-triggered control techniques for nonlinear systems. 

The foundations provided by this work shall allow dealing with many other scenarios including, when the plant (discrete or continuous) include time varying parameters, when the full state measurements are not available, or when policies other than the zero-order hold is used. Additionally, examining the various lifting techniques available in the literature is important, as well as, testing the scalability of the solution.

\appendix
\section{Lyapunov Stability Analysis}
This subsection of the appendix presents the exponential stability of the system given by
\begin{align}
    z_{k+1} = Z_1 L z_k + Z_1 N e_k
\end{align}
in a Lyapunov sense with the candidate $ V(k) = z_k^\top S z_k$. 

First, \( V(z_{k+1}) \) can be computed as follows
   \begin{equation}
    \begin{aligned}
       V(z_{k+1}) & = \left( Z_1 L z_k + Z_1 N e_k \right)^\top S \left( Z_1 L z_k + Z_1 N e_k \right) \\
       & = z_k^\top L^\top Z_1^\top S Z_1 L z_k + z_k^\top L^\top Z_1^\top S Z_1 N e_k + \dots \\
       \dots &~ e_k^\top N^\top Z_1^\top S Z_1 L z_k + e_k^\top N^\top Z_1^\top S Z_1 N e_k
    \end{aligned}
   \end{equation}
Lyapunov exponential stability condition with convergence rate $\alpha$ can be reached by defining
\( V(z_{k+1}) \leq \alpha V(z_k) \). This leads to the following identity based on the candidate Lyapunov function 
\begin{equation}
\label{eqn:haha}
    \begin{aligned}
        & z_k^\top L^\top Z_1^\top S Z_1 L z_k + z_k^\top L^\top Z_1^\top S Z_1 N e_k + \dots \\
        & \dots e_k^\top N^\top Z_1^\top S Z_1 L z_k + e_k^\top N^\top Z_1^\top S Z_1 N e_k \leq \alpha z_k^\top S z_k
    \end{aligned}
\end{equation}
By defining $v = \begin{bsmallmatrix}
    z_k \\ e_k
\end{bsmallmatrix}$, Eqn. (\ref{eqn:haha}) can be written in the form of \( v^\top \Psi v \leq 0 \), where
\begin{align}
       \Psi = \begin{bmatrix} 
   L^\top Z_1^\top S Z_1 L - \alpha S & L^\top Z_1^\top S Z_1 N \\ 
   N^\top Z_1^\top S Z_1 L & N^\top Z_1^\top S Z_1 N 
   \end{bmatrix}
\end{align}
Therefore, the Lyapunov stability condition for the system can be written as
\begin{align}
\label{eqn:cond}
    \begin{bmatrix} z_k \\ e_k \end{bmatrix}^{\top} \begin{bmatrix} 
   L^\top Z_1^\top S Z_1 L - \alpha S & L^\top Z_1^\top S Z_1 N \\ 
   N^\top Z_1^\top S Z_1 L & N^\top Z_1^\top S Z_1 N  \end{bmatrix} \begin{bmatrix} z_k \\ e_k \end{bmatrix} \leq 0
\end{align}
If condition (\ref{eqn:cond}) is satisfied, it then guarantees exponential stability of the system with convergence rate \(\alpha\).

\bibliography{references}

\begin{thebibliography}{59}
\providecommand{\natexlab}[1]{#1}
\providecommand{\url}[1]{\texttt{#1}}
\expandafter\ifx\csname urlstyle\endcsname\relax
  \providecommand{\doi}[1]{doi: #1}\else
  \providecommand{\doi}{doi: \begingroup \urlstyle{rm}\Url}\fi

\bibitem[{\AA}arz{\'e}n(1999)]{aaarzen1999simple}
Karl-Erik {\AA}arz{\'e}n.
\newblock {A simple event-based PID controller}.
\newblock \emph{IFAC Proceedings Volumes}, 32\penalty0 (2):\penalty0 8687--8692, 1999.

\bibitem[Eker et~al.(2000)Eker, Hagander, and {\AA}rz{\'e}n]{eker2000feedback}
Johan Eker, Per Hagander, and Karl-Erik {\AA}rz{\'e}n.
\newblock A feedback scheduler for real-time controller tasks.
\newblock \emph{Control Engineering Practice}, 8\penalty0 (12):\penalty0 1369--1378, 2000.

\bibitem[Tabuada(2007)]{tabuada2007event}
Paulo Tabuada.
\newblock Event-triggered real-time scheduling of stabilizing control tasks.
\newblock \emph{IEEE Transactions on Automatic control}, 52\penalty0 (9):\penalty0 1680--1685, 2007.

\bibitem[Heemels et~al.(2012)Heemels, Johansson, and Tabuada]{heemels2012introduction}
Wilhelmus~PMH Heemels, Karl~Henrik Johansson, and Paulo Tabuada.
\newblock An introduction to event-triggered and self-triggered control.
\newblock In \emph{2012 ieee 51st ieee conference on decision and control (cdc)}, pages 3270--3285. IEEE, 2012.

\bibitem[da~Silva et~al.(2018)da~Silva, Bazanella, Lorenzini, and Campestrini]{da2018data}
Gustavo R~Gon{\c{c}}alves da~Silva, Alexandre~S Bazanella, Charles Lorenzini, and Luciola Campestrini.
\newblock {Data-driven LQR control design}.
\newblock \emph{IEEE control systems letters}, 3\penalty0 (1):\penalty0 180--185, 2018.

\bibitem[{De Persis} and Tesi(2019)]{de2019formulas}
Claudio {De Persis} and Pietro Tesi.
\newblock Formulas for data-driven control: Stabilization, optimality, and robustness.
\newblock \emph{IEEE Transactions on Automatic Control}, 65\penalty0 (3):\penalty0 909--924, 2019.

\bibitem[Shi and Karydis(2021{\natexlab{a}})]{shi2021enhancement}
Lu~Shi and Konstantinos Karydis.
\newblock Enhancement for robustness of koopman operator-based data-driven mobile robotic systems.
\newblock In \emph{2021 IEEE International Conference on Robotics and Automation (ICRA)}, pages 2503--2510. IEEE, 2021{\natexlab{a}}.

\bibitem[Folkestad et~al.(2021)Folkestad, Wei, and Burdick]{folkestad2021quadrotor}
Carl Folkestad, Skylar~X Wei, and Joel~W Burdick.
\newblock Quadrotor trajectory tracking with learned dynamics: Joint koopman-based learning of system models and function dictionaries.
\newblock \emph{arXiv preprint arXiv:2110.10341}, 2021.

\bibitem[Susuki et~al.(2016)Susuki, Mezic, Raak, and Hikihara]{susuki2016applied}
Yoshihiko Susuki, Igor Mezic, Fredrik Raak, and Takashi Hikihara.
\newblock Applied koopman operator theory for power systems technology.
\newblock \emph{Nonlinear Theory and Its Applications, IEICE}, 7\penalty0 (4):\penalty0 430--459, 2016.

\bibitem[Brunton et~al.(2016{\natexlab{a}})Brunton, Proctor, and Kutz]{brunton2016discovering}
Steven~L Brunton, Joshua~L Proctor, and J~Nathan Kutz.
\newblock Discovering governing equations from data by sparse identification of nonlinear dynamical systems.
\newblock \emph{Proceedings of the national academy of sciences}, 113\penalty0 (15):\penalty0 3932--3937, 2016{\natexlab{a}}.

\bibitem[Manaa et~al.(2024{\natexlab{a}})Manaa, Elbalshy, and Abdallah]{manaa2024data}
Zeyad~M Manaa, Mohammed~R Elbalshy, and Ayman~M Abdallah.
\newblock Data-driven discovery of the quadrotor equations of motion via sparse identification of nonlinear dynamics.
\newblock In \emph{AIAA SCITECH 2024 Forum}, page 1308, 2024{\natexlab{a}}.

\bibitem[Jiang et~al.(2021)Jiang, Xiong, Zhang, Wang, Li, and Du]{jiang_modeling_2021}
Yu~Xin Jiang, Xiong Xiong, Shuo Zhang, Jia~Xiang Wang, Jia~Chun Li, and Lin Du.
\newblock Modeling and prediction of the transmission dynamics of {COVID}-19 based on the {SINDy}-{LM} method.
\newblock \emph{Nonlinear Dynamics}, 105:\penalty0 2775--2794, 2021.
\newblock ISSN 1573269X.
\newblock Number: 3.

\bibitem[Sorokina et~al.(2016)Sorokina, Sygletos, and Turitsyn]{sorokina_sparse_2016}
Mariia Sorokina, Stylianos Sygletos, and Sergei Turitsyn.
\newblock Sparse {Identification} for {Nonlinear} {Optical} {Communication} {Systems}: {SINO} {Method}.
\newblock \emph{Opt. Express}, 24\penalty0 (26):\penalty0 30433, December 2016.
\newblock ISSN 1094-4087.
\newblock \doi{10.1364/OE.24.030433}.

\bibitem[Bhadriraju et~al.(2020)Bhadriraju, Bangi, Narasingam, and Kwon]{bhadriraju_operable_2020}
Bhavana Bhadriraju, Mohammed Saad~Faizan Bangi, Abhinav Narasingam, and Joseph Sang~Il Kwon.
\newblock Operable adaptive sparse identification of systems: {Application} to chemical processes.
\newblock \emph{AIChE Journal}, 66\penalty0 (11), 2020.
\newblock ISSN 15475905.
\newblock \doi{10.1002/aic.16980}.
\newblock Number: 11.

\bibitem[Bhattacharya et~al.(2020)Bhattacharya, Cheng, and Xu]{bhattacharya_sparse_2020}
Dipankar Bhattacharya, Leo~K. Cheng, and Weiliang Xu.
\newblock Sparse {Machine} {Learning} {Discovery} of {Dynamic} {Differential} {Equation} of an {Esophageal} {Swallowing} {Robot}.
\newblock \emph{IEEE Transactions on Industrial Electronics}, 67\penalty0 (6):\penalty0 4711--4720, 2020.
\newblock ISSN 15579948.
\newblock \doi{10.1109/TIE.2019.2928239}.
\newblock Number: 6.

\bibitem[Sznaier(2020)]{sznaier2020control}
Mario Sznaier.
\newblock Control oriented learning in the era of big data.
\newblock \emph{IEEE Control Systems Letters}, 5\penalty0 (6):\penalty0 1855--1867, 2020.

\bibitem[Campi et~al.(2002)Campi, Lecchini, and Savaresi]{campi2002virtual}
Marco~C Campi, Andrea Lecchini, and Sergio~M Savaresi.
\newblock Virtual reference feedback tuning: a direct method for the design of feedback controllers.
\newblock \emph{Automatica}, 38\penalty0 (8):\penalty0 1337--1346, 2002.

\bibitem[Fliess and Join(2013)]{fliess2013model}
Michel Fliess and C{\'e}dric Join.
\newblock Model-free control.
\newblock \emph{International journal of control}, 86\penalty0 (12):\penalty0 2228--2252, 2013.

\bibitem[Liu et~al.(2023)Liu, Sun, Wang, Bullo, and Chen]{liu2023data}
Wenjie Liu, Jian Sun, Gang Wang, Francesco Bullo, and Jie Chen.
\newblock Data-driven self-triggered control via trajectory prediction.
\newblock \emph{IEEE Transactions on Automatic Control}, 68\penalty0 (11):\penalty0 6951--6958, 2023.

\bibitem[Digge and Pasumarthy(2022)]{digge2022data}
Vijayanand Digge and Ramkrishna Pasumarthy.
\newblock {Data-driven event-triggered control for discrete-time LTI systems}.
\newblock In \emph{2022 European Control Conference (ECC)}, pages 1355--1360. IEEE, 2022.

\bibitem[Wang et~al.(2023)Wang, Berberich, Sun, Wang, Allg{\"o}wer, and Chen]{wang2023model}
Xin Wang, Julian Berberich, Jian Sun, Gang Wang, Frank Allg{\"o}wer, and Jie Chen.
\newblock Model-based and data-driven control of event-and self-triggered discrete-time linear systems.
\newblock \emph{IEEE Transactions on Cybernetics}, 2023.

\bibitem[Koopman(1931)]{koopman1931hamiltonian}
Bernard~O Koopman.
\newblock Hamiltonian systems and transformation in hilbert space.
\newblock \emph{Proceedings of the National Academy of Sciences}, 17\penalty0 (5):\penalty0 315--318, 1931.

\bibitem[Koopman and Neumann(1932)]{koopman1932dynamical}
Bernard~O Koopman and J~v Neumann.
\newblock Dynamical systems of continuous spectra.
\newblock \emph{Proceedings of the National Academy of Sciences}, 18\penalty0 (3):\penalty0 255--263, 1932.

\bibitem[Mezi{\'c} and Banaszuk(2004)]{mezic2004comparison}
Igor Mezi{\'c} and Andrzej Banaszuk.
\newblock Comparison of systems with complex behavior.
\newblock \emph{Physica D: Nonlinear Phenomena}, 197\penalty0 (1-2):\penalty0 101--133, 2004.

\bibitem[Mezi{\'c}(2005)]{mezic2005spectral}
Igor Mezi{\'c}.
\newblock Spectral properties of dynamical systems, model reduction and decompositions.
\newblock \emph{Nonlinear Dynamics}, 41:\penalty0 309--325, 2005.

\bibitem[Bruder et~al.(2021)Bruder, Fu, Gillespie, Remy, and Vasudevan]{bruder2021koopman}
Daniel Bruder, Xun Fu, R~Brent Gillespie, C~David Remy, and Ram Vasudevan.
\newblock Koopman-based control of a soft continuum manipulator under variable loading conditions.
\newblock \emph{IEEE robotics and automation letters}, 6\penalty0 (4):\penalty0 6852--6859, 2021.

\bibitem[Mamakoukas et~al.(2021)Mamakoukas, Castano, Tan, and Murphey]{mamakoukas2021derivative}
Giorgos Mamakoukas, Maria~L Castano, Xiaobo Tan, and Todd~D Murphey.
\newblock Derivative-based koopman operators for real-time control of robotic systems.
\newblock \emph{IEEE Transactions on Robotics}, 37\penalty0 (6):\penalty0 2173--2192, 2021.

\bibitem[Zhu et~al.(2022)Zhu, Ding, Jia, and Feng]{zhu2022koopman}
Xuehong Zhu, Chengjun Ding, Lizhen Jia, and Yubo Feng.
\newblock Koopman operator based model predictive control for trajectory tracking of an omnidirectional mobile manipulator.
\newblock \emph{Measurement and Control}, 55\penalty0 (9-10):\penalty0 1067--1077, 2022.

\bibitem[Komeno et~al.(2022)Komeno, Michael, K{\"u}chler, Anarossi, and Matsubara]{komeno2022deep}
Naoto Komeno, Brendan Michael, Katharina K{\"u}chler, Edgar Anarossi, and Takamitsu Matsubara.
\newblock Deep koopman with control: Spectral analysis of soft robot dynamics.
\newblock In \emph{2022 61st Annual Conference of the Society of Instrument and Control Engineers (SICE)}, pages 333--340. IEEE, 2022.

\bibitem[Han et~al.(2021)Han, Euler-Rolle, and Katzschmann]{han2021desko}
Minghao Han, Jacob Euler-Rolle, and Robert~K Katzschmann.
\newblock Desko: Stability-assured robust control with a deep stochastic koopman operator.
\newblock In \emph{International Conference on Learning Representations}, 2021.

\bibitem[Manaa et~al.(2024{\natexlab{b}})Manaa, Abdallah, Abido, and Ali]{manaa2024koopman}
Zeyad~M Manaa, Ayman~M Abdallah, Mohammad~A Abido, and Syed S~Azhar Ali.
\newblock {Koopman-LQR Controller for Quadrotor UAVs from Data}.
\newblock \emph{arXiv preprint arXiv:2406.17973}, 2024{\natexlab{b}}.

\bibitem[Hossain et~al.(2023)Hossain, Adesunkanmi, and Kumar]{hossain2023data}
Ramij~Raja Hossain, Rahmat Adesunkanmi, and Ratnesh Kumar.
\newblock Data-driven linear koopman embedding for networked systems: Model-predictive grid control.
\newblock \emph{IEEE Systems Journal}, 2023.

\bibitem[Markmann et~al.(2024)Markmann, Straat, and Hammer]{markmann2024koopman}
Thorben Markmann, Michiel Straat, and Barbara Hammer.
\newblock Koopman-based surrogate modelling of turbulent rayleigh-b$\backslash$'enard convection.
\newblock \emph{arXiv preprint arXiv:2405.06425}, 2024.

\bibitem[Mezi{\'c} et~al.(2024)Mezi{\'c}, Drma{\v{c}}, {\v{C}}rnjari{\'c}, Ma{\'c}e{\v{s}}i{\'c}, Fonoberova, Mohr, Avila, Manojlovi{\'c}, and Andrej{\v{c}}uk]{mezic2024koopman}
Igor Mezi{\'c}, Zlatko Drma{\v{c}}, Nelida {\v{C}}rnjari{\'c}, Senka Ma{\'c}e{\v{s}}i{\'c}, Maria Fonoberova, Ryan Mohr, Allan~M Avila, Iva Manojlovi{\'c}, and Aleksandr Andrej{\v{c}}uk.
\newblock A koopman operator-based prediction algorithm and its application to covid-19 pandemic and influenza cases.
\newblock \emph{Scientific reports}, 14\penalty0 (1):\penalty0 5788, 2024.

\bibitem[Willems et~al.(2005)Willems, Rapisarda, Markovsky, and De~Moor]{willems2005note}
Jan~C Willems, Paolo Rapisarda, Ivan Markovsky, and Bart~LM De~Moor.
\newblock A note on persistency of excitation.
\newblock \emph{Systems \& Control Letters}, 54\penalty0 (4):\penalty0 325--329, 2005.

\bibitem[Budi{\v{s}}i{\'c} et~al.(2012)Budi{\v{s}}i{\'c}, Mohr, and Mezi{\'c}]{budivsic2012applied}
Marko Budi{\v{s}}i{\'c}, Ryan Mohr, and Igor Mezi{\'c}.
\newblock Applied koopmanism.
\newblock \emph{Chaos: An Interdisciplinary Journal of Nonlinear Science}, 22\penalty0 (4), 2012.

\bibitem[Bevanda et~al.(2021)Bevanda, Sosnowski, and Hirche]{bevanda2021koopman}
Petar Bevanda, Stefan Sosnowski, and Sandra Hirche.
\newblock Koopman operator dynamical models: Learning, analysis and control.
\newblock \emph{Annual Reviews in Control}, 52:\penalty0 197--212, 2021.

\bibitem[Proctor et~al.(2018)Proctor, Brunton, and Kutz]{proctor2018generalizing}
Joshua~L Proctor, Steven~L Brunton, and J~Nathan Kutz.
\newblock Generalizing koopman theory to allow for inputs and control.
\newblock \emph{SIAM Journal on Applied Dynamical Systems}, 17\penalty0 (1):\penalty0 909--930, 2018.

\bibitem[Peitz et~al.(2020)Peitz, Otto, and Rowley]{peitz2020data}
Sebastian Peitz, Samuel~E Otto, and Clarence~W Rowley.
\newblock Data-driven model predictive control using interpolated koopman generators.
\newblock \emph{SIAM Journal on Applied Dynamical Systems}, 19\penalty0 (3):\penalty0 2162--2193, 2020.

\bibitem[Korda and Mezi{\'c}(2018)]{korda2018linear}
Milan Korda and Igor Mezi{\'c}.
\newblock Linear predictors for nonlinear dynamical systems: Koopman operator meets model predictive control.
\newblock \emph{Automatica}, 93:\penalty0 149--160, 2018.

\bibitem[Otto and Rowley(2019)]{otto2019linearly}
Samuel~E Otto and Clarence~W Rowley.
\newblock Linearly recurrent autoencoder networks for learning dynamics.
\newblock \emph{SIAM Journal on Applied Dynamical Systems}, 18\penalty0 (1):\penalty0 558--593, 2019.

\bibitem[Yeung et~al.(2019)Yeung, Kundu, and Hodas]{yeung2019learning}
Enoch Yeung, Soumya Kundu, and Nathan Hodas.
\newblock Learning deep neural network representations for koopman operators of nonlinear dynamical systems.
\newblock In \emph{2019 American Control Conference (ACC)}, pages 4832--4839. IEEE, 2019.

\bibitem[Lusch et~al.(2018)Lusch, Kutz, and Brunton]{lusch2018deep}
Bethany Lusch, J~Nathan Kutz, and Steven~L Brunton.
\newblock Deep learning for universal linear embeddings of nonlinear dynamics.
\newblock \emph{Nature communications}, 9\penalty0 (1):\penalty0 4950, 2018.

\bibitem[Netto et~al.(2021)Netto, Susuki, Krishnan, and Zhang]{netto2021analytical}
Marcos Netto, Yoshihiko Susuki, Venkat Krishnan, and Yingchen Zhang.
\newblock On analytical construction of observable functions in extended dynamic mode decomposition for nonlinear estimation and prediction.
\newblock In \emph{2021 American Control Conference (ACC)}, pages 4190--4195. IEEE, 2021.

\bibitem[Kamb et~al.(2020)Kamb, Kaiser, Brunton, and Kutz]{kamb2020time}
Mason Kamb, Eurika Kaiser, Steven~L Brunton, and J~Nathan Kutz.
\newblock Time-delay observables for koopman: Theory and applications.
\newblock \emph{SIAM Journal on Applied Dynamical Systems}, 19\penalty0 (2):\penalty0 886--917, 2020.

\bibitem[Mamakoukas et~al.(2019)Mamakoukas, Castano, Tan, and Murphey]{mamakoukas2019local}
Giorgos Mamakoukas, Maria Castano, Xiaobo Tan, and Todd Murphey.
\newblock Local koopman operators for data-driven control of robotic systems.
\newblock In \emph{Robotics: science and systems}, 2019.

\bibitem[Schmid(2010)]{schmid2010dynamic}
Peter~J Schmid.
\newblock Dynamic mode decomposition of numerical and experimental data.
\newblock \emph{Journal of fluid mechanics}, 656:\penalty0 5--28, 2010.

\bibitem[Proctor et~al.(2016)Proctor, Brunton, and Kutz]{proctor2016dynamic}
Joshua~L Proctor, Steven~L Brunton, and J~Nathan Kutz.
\newblock Dynamic mode decomposition with control.
\newblock \emph{SIAM Journal on Applied Dynamical Systems}, 15\penalty0 (1):\penalty0 142--161, 2016.

\bibitem[{De Persis} et~al.(2023){De Persis}, Postoyan, and Tesi]{de2023event}
Claudio {De Persis}, Romain Postoyan, and Pietro Tesi.
\newblock Event-triggered control from data.
\newblock \emph{IEEE Transactions on Automatic Control}, 2023.

\bibitem[Boyd and Vandenberghe(2004)]{boyd2004convex}
Stephen~P Boyd and Lieven Vandenberghe.
\newblock \emph{Convex optimization}.
\newblock Cambridge university press, 2004.

\bibitem[Brunton et~al.(2016{\natexlab{b}})Brunton, Brunton, Proctor, and Kutz]{brunton2016koopman}
Steven~L Brunton, Bingni~W Brunton, Joshua~L Proctor, and J~Nathan Kutz.
\newblock Koopman invariant subspaces and finite linear representations of nonlinear dynamical systems for control.
\newblock \emph{PloS one}, 11\penalty0 (2), 2016{\natexlab{b}}.

\bibitem[Surana and Banaszuk(2016)]{surana2016linear}
Amit Surana and Andrzej Banaszuk.
\newblock Linear observer synthesis for nonlinear systems using koopman operator framework.
\newblock \emph{IFAC-PapersOnLine}, 49\penalty0 (18):\penalty0 716--723, 2016.

\bibitem[Surana(2016)]{surana2016koopman}
Amit Surana.
\newblock Koopman operator based observer synthesis for control-affine nonlinear systems.
\newblock In \emph{2016 IEEE 55th Conference on Decision and Control}, pages 6492--6499. IEEE, 2016.

\bibitem[Do et~al.(2024)Do, Uchytil, and Hur{\'a}k]{do2024practical}
Loi Do, Adam Uchytil, and Zden{\v{e}}k Hur{\'a}k.
\newblock Practical guidelines for data-driven identification of lifted linear predictors for control.
\newblock \emph{arXiv preprint arXiv:2408.01116}, 2024.

\bibitem[Williams et~al.(2015)Williams, Kevrekidis, and Rowley]{williams2015data}
Matthew~O Williams, Ioannis~G Kevrekidis, and Clarence~W Rowley.
\newblock A data--driven approximation of the koopman operator: Extending dynamic mode decomposition.
\newblock \emph{Journal of Nonlinear Science}, 25:\penalty0 1307--1346, 2015.

\bibitem[Ge et~al.(2024)Ge, Xu, and Wu]{ge2024deep}
Jiacheng Ge, Yijun Xu, and Zaijun Wu.
\newblock Deep learning-based construction of koopman observable functions for power system nonlinear dynamics.
\newblock In \emph{2024 IEEE Power \& Energy Society General Meeting (PESGM)}, pages 1--5. IEEE, 2024.

\bibitem[Shi and Karydis(2021{\natexlab{b}})]{shi2021acd}
Lu~Shi and Konstantinos Karydis.
\newblock Acd-edmd: Analytical construction for dictionaries of lifting functions in koopman operator-based nonlinear robotic systems.
\newblock \emph{IEEE Robotics and Automation Letters}, 7\penalty0 (2):\penalty0 906--913, 2021{\natexlab{b}}.

\bibitem[Williams et~al.(2014)Williams, Rowley, and Kevrekidis]{williams2014kernel}
Matthew~O Williams, Clarence~W Rowley, and Ioannis~G Kevrekidis.
\newblock A kernel-based approach to data-driven koopman spectral analysis.
\newblock \emph{arXiv preprint arXiv:1411.2260}, 2014.

\bibitem[Lee et~al.(2024)Lee, Hamzi, Hou, Owhadi, Santin, and Vaidya]{lee2024kernel}
Jonghyeon Lee, Boumediene Hamzi, Boya Hou, Houman Owhadi, Gabriele Santin, and Umesh Vaidya.
\newblock Kernel methods for the approximation of the eigenfunctions of the koopman operator.
\newblock \emph{arXiv preprint arXiv:2412.16588}, 2024.

\end{thebibliography}

\end{document}